\documentclass{cccg17}
\usepackage{graphicx,amssymb,amsmath}
\usepackage{graphicx}
\usepackage{graphics}
\usepackage{epsfig}
\usepackage{epstopdf}
\usepackage{amsmath}
\usepackage{latexsym}
\usepackage{subfig}
\usepackage{wrapfig}
\usepackage{multirow}
\usepackage{amsfonts}
\usepackage{cite}
\usepackage{amssymb}
\usepackage{caption}
\usepackage{algorithm}
\usepackage{algorithmic}

\makeatletter
\newcommand{\rmnum}[1]{\romannumeral #1}
\newcommand{\Rmnum}[1]{\expandafter\@slowromancap\romannumeral #1@}
\makeatother

\title{Balanced $k$-Center Clustering When $k$ Is A Constant}

\author{Hu Ding\thanks{Department of Computer Science and Engineering,
        Michigan State University, {\tt huding@msu.edu}}
        }

\index{Gumble, Barney}
\index{Szyslak, Moe}

\begin{document}
\thispagestyle{empty}
\maketitle

\begin{abstract}
The problem of constrained $k$-center clustering has attracted significant attention in the past decades. In this paper, we study balanced $k$-center cluster where the size of each cluster is constrained by the given lower and upper bounds. The problem is motivated by the applications in processing and analyzing large-scale data in high dimension. We provide a simple nearly linear time $4$-approximation algorithm when the number of clusters $k$ is assumed to be a constant. Comparing with existing method, our algorithm improves the approximation ratio and significantly reduces the time complexity. Moreover, our result can be easily extended to any metric space.

\end{abstract}

 \section{Introduction}
\label{sec-set}
The $k$-center clustering is a fundamental problem in computer science and has numerous applications in real world. Given a set of points in Euclidean space and a positive integer $k$, the problem seeks $k$ balls to cover all the points such that  the maximum radius of the balls is minimized. Another variant of $k$-center clustering considers the case that all the given points (vertices) form a metric graph and the centers of the balls are chosen from the vertices. The optimal approximation results appeared in the 80's: Gonazlez~\cite{G85} and Hochbaum and Shmoys~\cite{HS85} provided a 2-approximation and proved that any approximation ratio $c<2$ would imply $P=NP$. Besides the classic problem, several variants of $k$-center clustering with upper~\cite{BKP,KS00,CHK,ABC,KC14} or lower~\cite{APF,EHR,AS16} bounds on cluster sizes have been extensively studied in recent years. In particular, Ding et al.~\cite{DHH} studied $k$-center clustering with both upper and lower bounded cluster sizes which is also called {\em Balanced $k$-Center Clustering}. Most of existing methods model these constrained $k$-center clustering problems as linear integer programming and design novel rounding algorithms to obtain constant approximations.

Besides the well studied applications in data analysis and facility location, balanced $k$-center clustering is particularly motivated by the arising problems in {\em big data}~\cite{DLP,ABM,BBL}. For example, we need to dispatch data to multiple machines for processing if the data scale is extremely large; at the same time we have to consider the balancedness, because the machines receiving too much data could be the bottleneck of the system and the ones receiving too little data is not sufficiently energy-efficient. 

In this paper, we consider the balanced $k$-center clustering problem in high dimension and assume that $k$ is a constant. The rationale for the assumption is twofold: $k$ is usually not large in practice (e.g., the data is distributed over less than $10$ machines); even if $k$ is large, we can first partition the data into multiple groups and perform balanced $k$-center clustering for each group with a much smaller $k$ (similar to the manner of {\em hierarchical  clustering}~\cite{KR90}).

\noindent\textbf{Our main result.} Given an instance of $k$-center clustering with upper and lower bounds on cluster sizes, we develop a nearly linear time $4$-approximation algorithm. We assume that the dimensionality $d$ is large and the number of clusters $k$ is a constant. The key techniques contains two parts. First, we observe that Gonazlez's algorithm~\cite{G85} could provide a set of candidates for the $k$ cluster centers and at least one candidate yields $4$-approximation (Lemma~\ref{lem-4app}). Secondly, we develop a novel rounding procedure to select the qualified candidate and generate a feasible solution for the balanced $k$-center clustering (Lemma~\ref{lem-integer}); note that a straightforward idea for the selection task is modeling it as a maximum flow problem but the running time would be at least quadratic. Comparing with the existing method for balanced $k$-center clustering~\cite{DHH}, we improve the approximation ratio from $6$ to $4$ and significantly reduce the running time via avoiding to solve the large-scale linear programming. 

Also, our result can be easily extended to any metric space and the running time depends on the complexity for acquiring the distance between any two points (e.g., the complexity is $O(d)$ in Euclidean space).

\noindent\textbf{Notation.} Throughout the paper we denote the input as a set of $n$ points $P$ in $\mathbb{R}^d$ and an integer $k\geq 1$; we further constrain the size of each cluster by the lower and upper bounds $L$ and $U \in \mathbb{Z}^+$ (to ensure that a feasible solution exists, we assume $1\leq L\leq \lfloor\frac{n}{k}\rfloor \leq \lceil\frac{n}{k}\rceil\leq U\leq n$).

%

\section{Our Algorithm}
\label{sec-high}

\subsection{Finding The Candidates For Cluster Centers}
\label{sec-candidate}

Gonazlez's seminal paper~\cite{G85} provided a very simple $2$-approximation algorithm for $k$-center clustering in any dimension. Basically, the algorithm iteratively selects $k$ points from the input, where the initial point is arbitrarily selected, and each following $j$-th step ($2\leq j\leq k$) chooses the point which has the largest minimum distance to the already selected $j-1$ points. Finally, it is able to show that these $k$ points induce a $2$-approximation for $k$-center clustering if each input point is assigned to its nearest neighbor of these $k$ points. 

We denote these $k$ points selected by Gonazlez's algorithm as $S=\{s_1, s_2, \cdots, s_k\}$, and define the set $S^k$ as $\underbrace{S\times\cdots\times S}_k$, i.e., $\{(s'_1, s'_2, \cdots, s'_k)\mid s'_j\in S, 1\leq j\leq k\}$. Then we have the following lemma. 

\begin{lemma}
\label{lem-4app}
There exists a $k$-tuple points from $S^k$ yielding a $4$-approximation for balanced $k$-center clustering. 
\end{lemma}
\begin{proof}
Suppose the unknown $k$ optimal balanced clusters are $C_1, C_2, \cdots, C_k$, and the optimal radius is $r_{opt}$. If the selected $k$ points of $S$ luckily fall to these $k$ clusters separately, it is easy to obtain a $2$-approximation via triangle inequality (we will discuss that how to assign the input points to the $k$ cluster centers for satisfying the requirement of balance in Section~\ref{sec-feasible}). 

Now, we consider the other case.  Without loss of generality, we assume that $s_{j_1}$ and $s_{j_2}$ is the firstly appeared pair belonging to a same optimal cluster and $j_1<j_2$. For the sake of simplicity, we assume that $s_j\in C_j$ for $1\leq j\leq j_2-1$. Due to the nature of Gonazlez's algorithm, we know that 
\begin{eqnarray}
\max_{p\in \cup^k_{j=j_2}C_j}\{\min_{1\leq l\leq j_2-1}||p-s_l||\}\leq ||s_{j_1}-s_{j_2}||\leq 2r_{opt}.\label{for-4app1}
\end{eqnarray}
Note even for the points from a same cluster $C_j$ where $j\geq j_2$, their nearest neighbors from $\{s_1, \cdots, s_{j_2-1}\}$ are not necessarily same. Moreover, because of the requirement of balance, we cannot simply assign them to their nearest neighbors to generate a $2$-approximation by (\ref{for-4app1}); actually, this is also the major difference between ordinary and balanced $k$-center clustering. 
Instead, for each $j\geq j_2$ we arbitrarily select a point $p\in C_j$, and assign the whole $C_j$ to $p$'s nearest neighbor of $\{s_1, \cdots, s_{j_2-1}\}$ which is denoted as $s_{l(j)}$. Correspondingly, for any $p'\in C_j$ we have
\begin{eqnarray}
||p'-s_{l(j)}||\leq ||p'-p||+||p-s_{l(j)}||\leq 4r_{opt} \label{for-4app2}
\end{eqnarray}
due to triangle inequality and the fact that both $||p'-p||$ and $||p-s_{l(j)}||$ are no larger than $2r_{opt}$. Thus, the $k$-tuple points $\{s_1,$ $ s_2,$ $ \cdots,$ $ s_{j_2-1},$ $ s_{l(j_2)},$ $ \cdots, s_{l(k)}\}$ yields a $4$-approximation if each optimal cluster $C_j$ is assigned to the $j$-th point in the tuple for $1\leq j\leq k$. 
\end{proof}

\subsection{Finding A Feasible Solution}
\label{sec-feasible}

Next, we answer the question that how to assign the input points to a fixed $k$-tuple points to satisfy the requirement of balance. To show its generalization, we denote the given $k$-tuple as $\{q_1, q_2, \cdots, q_k\}$ which is not necessarily from $S^k$. It is easy to know that the qualified radii must come from the $kn$ distances $\{||p-q_j||\mid p\in P, 1\leq j\leq k\}$. As a consequence, we can apply binary search to find the smallest qualified radius. For each candidate radius $r$, we draw $k$ balls centered at the $k$-tuple points and with the radius $r$ respectively. We denote the $k$ balls as $\mathcal{B}_1, \cdots, \mathcal{B}_k$. Thus, the only remaining problem is determining that whether there exists a balanced clustering on $P$ to be covered by such $k$ balls. We call such a clustering as a {\em feasible solution} if it exists. 

A straightforward way to find a feasible solution is building a bipartite graph between the $n$ points of $P$ and the $k$ balls, where a point is connected to a ball if it is covered by the ball; each ball has a capacity $U$ and demand $L$, and the maximum flow from the points to balls is $n$ if and only if a feasible solution exists. The existing maximum flow algorithms, such as Ford-Fulkerson algorithm or the new Orlin's algorithm~\cite{O13}, costs at least $O(n^2)$ time. Recall that $k$ is constant, and below we will show that the problem can be solved by \textbf{a system of linear equations and inequalities (SoL)} with the size independent of $n$. 

The region $\cup^k_{j=1}\mathcal{B}_j$ divides the space into $2^k-1$ parts (we ignore the region outside the union of the balls, since no point locates there; otherwise, we can simply reject this candidate $r$).
We use $\mathcal{R}_{(j_1, j_2, \cdots, j_t)}$ with $1\leq j_1<j_2<\cdots<j_t\leq k$ to indicate the region 
$$(\mathcal{B}_{j_1}\cap \cdots\cap \mathcal{B}_{j_t})\setminus(\cup_{j\notin\{j_1, \cdots, j_t\}}\mathcal{B}_j).$$
We calculate the total number of points covered by $\mathcal{R}_{(j_1, j_2, \cdots, j_t)}$ which is denoted as $n_{(j_1, j_2, \cdots, j_t)}$, and assign $t$ non-negative variables $x^{j_1}_{(j_1, j_2, \cdots, j_t)}, \cdots, x^{j_t}_{(j_1, j_2, \cdots, j_t)}$ where each $x^{j_l}_{(j_1, j_2, \cdots, j_t)}$ indicates the number of points assigned to the $j_l$-th cluster from $\mathcal{R}_{(j_1, j_2, \cdots, j_t)}$. Thus, we have the following two types of linear constraint.
\begin{eqnarray}
x^{j_1}_{(j_1, j_2, \cdots, j_t)}+ \cdots +x^{j_t}_{(j_1, j_2, \cdots, j_t)}=n_{(j_1, j_2, \cdots, j_t)}, \label{for-lp1}\\
L\leq \sum_{(j_1, j_2, \cdots, j_t)\in \pi_{j_l}} x^{j_l}_{(j_1, j_2, \cdots, j_t)}\leq U. \label{for-lp2}
\end{eqnarray}
Here $\pi_{j_l}$ is the set of all the possible subsets containing $j_l$ of $\{1, \cdots, k\}$. There are at most $k2^{k}$ variables and $O(2^k)$ linear constraints in the whole SoL. Since $k$ is a constant, the time complexity for building such a SoL is $O(nd)$ which is dominated by computing the distances between the $n$ points and $k$ ball centers. Further, the time complexity for solving the SoL is $O(poly (2^k))$ via Gaussian elimination.

Once obtaining a feasible solution of the above SoL, we still need to check that whether the solution is an integer solution for generating a clustering result. 

\begin{lemma}
\label{lem-integer}
If there exists a feasible solution of the above SoL, we can always transform it to an integer solution in $O(poly (2^k))$ time. 
\end{lemma}
\begin{proof}

\begin{figure}[h]
\vspace{-.2cm} 
\centering
 \includegraphics[height=1.1in]{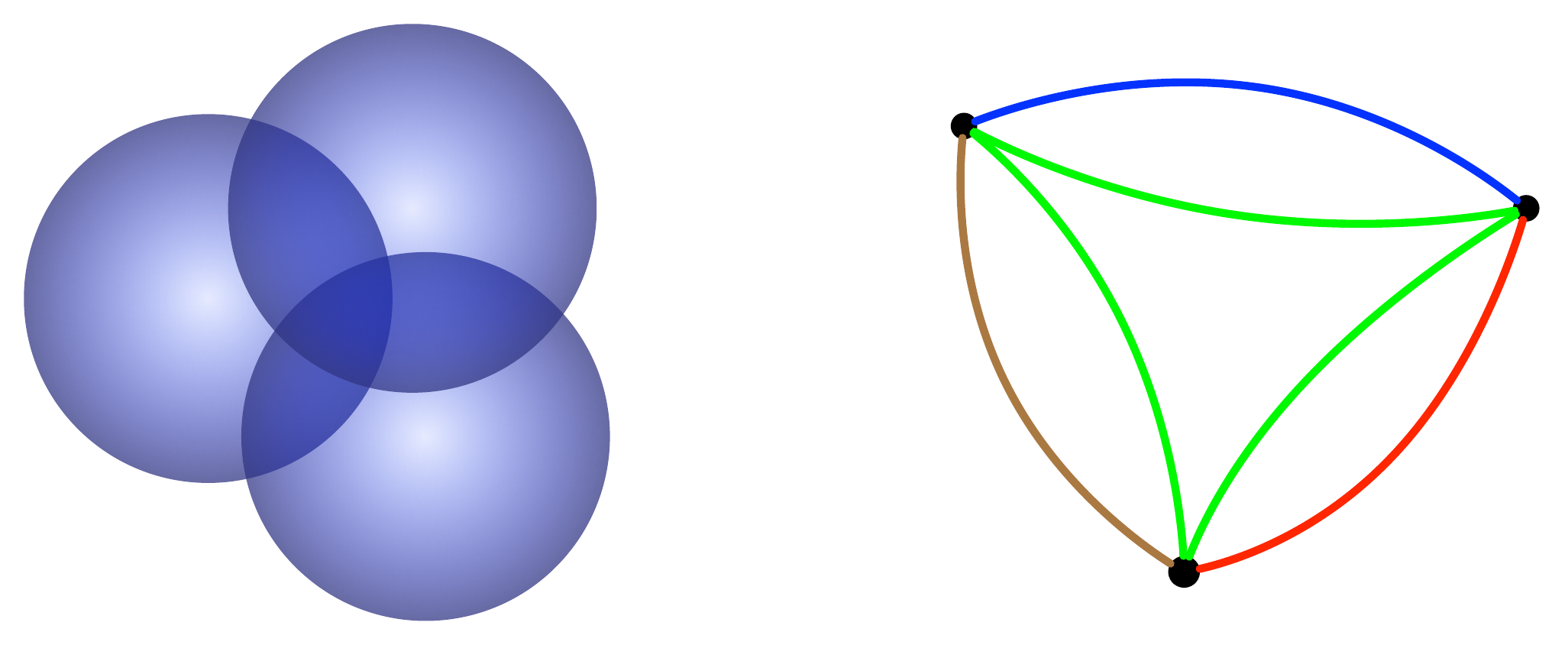}
  \vspace{-0.1in}
  \caption{An illustration for building the multigraph $G$. Suppose $k=3$ and the three balls locate as the left figure shows. For the sake of simplicity, we assume that all the variables corresponding to the overlapping areas are fractional. The colored multigraph $G$ is in the right. The green edges correspond to the intersection of the three balls; any two vertices have another individually colored edge corresponding to their own intersection.} 
   \label{fig-graph}
  \vspace{-0.15in}
\end{figure}

Suppose we have a fractional feasible solution denoted as $\Gamma=\{x^{j_l}_{(j_1, j_2, \cdots, j_t)}\mid j_l\in (j_1, j_2, \cdots, j_t),$ $ (j_1, j_2, \cdots, j_t)\in \pi\}$ where $\pi=2^{\{1, \cdots, k\}}$. To help our analysis, we also construct a colored multigraph $G(V,E)$, where $V$ contains $k$ vertices $\{v_j\mid 1\leq j\leq k\}$ corresponding to the $k$ balls $\{\mathcal{B}_j\mid 1\leq j\leq k\}$ respectively. Moreover, for any region $\mathcal{R}_{(j_1, j_2, \cdots, j_t)}$ and any pair $j_l, j_{l'}$ with $1\leq l<l'\leq t$, we add an edge between $v_{j_l}$ and $v_{j_{l'}}$ if both $x^{j_l}_{(j_1, j_2, \cdots, j_t)}$ and $x^{j_{l'}}_{(j_1, j_2, \cdots, j_t)}$ are fractional values. Thus, it is possible to have multiple edges between two vertices. Also, the edges corresponding to each $\mathcal{R}_{(j_1, j_2, \cdots, j_t)}$ share the same color (see Fig.~\ref{fig-graph}). Consider the following three cases.\\

\noindent\textbf{Case \Rmnum{1}.} If there is a circle having at least two different colors in $G$, we denote it as $v_{1}\rightarrow v_2\rightarrow\cdots\rightarrow v_h\rightarrow v_1$ w.l.o.g. From the construction of $G$, we know that for any two neighborhoods in the circle, there are two corresponding numbers from $\Gamma$ which are both fractional and share the same region. Let the couples of numbers be 
\begin{eqnarray}
(x^1_{*_1},x^2_{*_1}), (x^2_{*_2},x^3_{*_2}), \cdots, (x^h_{*_h},x^1_{*_h}).
\end{eqnarray}
Here we denote the foot subscripts by $*_j$ to simplify our analysis. If there exist two consecutive edges sharing the same color, e.g., $\overline{v_1 v_2}$ and $\overline{v_2 v_3}$, from the construction of $G$ we know that $v_1$ and $v_3$ are connected by an edge with the same color as well. Hence we can always delete $v_2$ from the circle and add the edge $\overline{v_1 v_3}$. Therefore we can assume that any neighbor edges have different colors in the circle, i.e., the following claim.\\

\noindent\textbf{Claim.}
{\em $*_{j-1}\neq *_j$ for $2\leq j\leq h$ and $*_h\neq *_1$. }\\

Meanwhile, we choose the small positive value 
\begin{eqnarray}
\delta=\min\{x^j_{*_j}-\lfloor x^j_{*_j}\rfloor, \lceil x^j_{*_{j-1}}\rceil -x^j_{*_{j-1}}\mid 1\leq j\leq h\}
\end{eqnarray}
where $x^1_{*_0}$ represents $x^1_{*_h}$ for convenience. Together with the above claim we know that the following numbers
\small{
\begin{eqnarray}
x^1_{*_1}-\delta,x^2_{*_1}+\delta, x^2_{*_2}-\delta,x^3_{*_2}+\delta, \cdots, x^h_{*_h}-\delta,x^1_{*_h}+\delta
\end{eqnarray}
}
\normalsize
contain at least one integer and all the others remain non-negative (see Fig.~\ref{fig-circle}). More importantly, no constraint of the SoL is violated after this adjustment. Since this operation adds new integers to $\Gamma$, we have to remove some edges of $G$ due to the rule of its construction. If we keep adjusting the fractional values of $\Gamma$ by this way, the edges of $G$ will become fewer and fewer. After finite steps, there will be no circle or each circle has only one color, i.e., one of the next two cases happens. Actually the following two cases can be handled by similar manners. In order to show our idea more clearly, we discuss the simpler one, Case~\Rmnum{2}, first. \\

\begin{figure}[h]
\vspace{-.1cm} 
\centering
 \includegraphics[height=0.7in]{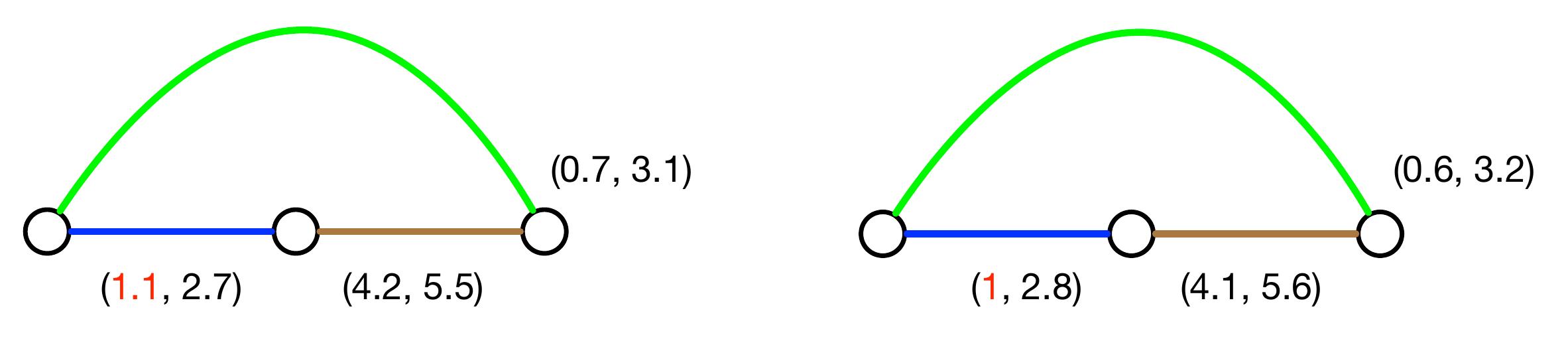}
  \vspace{-0.1in}
  \caption{An illustration for adjusting the fractional numbers for Case \Rmnum{1}. The left shows an example circle with $h=3$ and the couples of numbers. The right shows the couples of numbers after the adjustment with $\delta=0.1$.} 
   \label{fig-circle}
  \vspace{-0.15in}
\end{figure}

\noindent\textbf{Case \Rmnum{2}.} Now, we consider the second case that no circle exists in $G$; in other words, $G$ is a forest. Different to the first case, we arbitrarily pick a leaf-to-leaf path in $G$ and denote it as $v_{1}\rightarrow v_2\rightarrow\cdots\rightarrow v_h$, i.e., $v_1$ and $v_h$ are two leaves in $G$ (see Fig.~\ref{fig-tree}). Also from the construction of $G$, we have the following couples of fractional values
\begin{eqnarray}
(x^1_{*_1},x^2_{*_1}), (x^2_{*_2},x^3_{*_2}), \cdots, (x^{h-1}_{*_{h-1}},x^h_{*_{h-1}}).
\end{eqnarray}
Moreover, it is easy to know that $*_j\neq *_{j+1}$ for $1\leq j\leq h-2$; otherwise, there will be a circle $v_j\rightarrow v_{j+1}\rightarrow v_{j+2}\rightarrow v_{j}$ due to the construction of $G$ (which is contradict to the definition of Case \Rmnum{2}). Because $v_1$ is a leaf, we know that only one number of $\{x^1_{*}\mid *\in \pi_1\}$ is fractional and thus $\sum_{*\in \pi}x^1_*$ is fractional. Note that both $L$ and $U$ are integers, so the constraint (\ref{for-lp2}) is not tight in both sides for $j_l=1$, and similarly for $j_l=h$ too. We choose $\delta=\min\{x^j_{*_j}-\lfloor x^j_{*_j}\rfloor, \lceil x^{j+1}_{*_{j}}\rceil -x^{j+1}_{*_{j}}\mid 1\leq j\leq h-1\}$. Through the same manner for analyzing the first case, we know that the following numbers 
\begin{eqnarray}
x^1_{*_1}-\delta,x^2_{*_1}+\delta, x^2_{*_2}-\delta,x^3_{*_2}+\delta, \nonumber\\
\cdots, x^{h-1}_{*_{h-1}}-\delta,x^h_{*_{h-1}}+\delta
\end{eqnarray}
contain at least one integer and all the others remain non-negative, while no constraint of the SoL is violated. In particular, the constraint (\ref{for-lp2}) for $j_l=1$ and $h$ still holds since they are not tight before. Then we update $G$ by removing some edges. If we keep performing this adjustment finite times, $G$ will contain no edge. That is, we obtain an integer solution for the SoL.\\

\begin{figure}[h]
\vspace{-.1cm} 
\centering
 \includegraphics[height=0.9in]{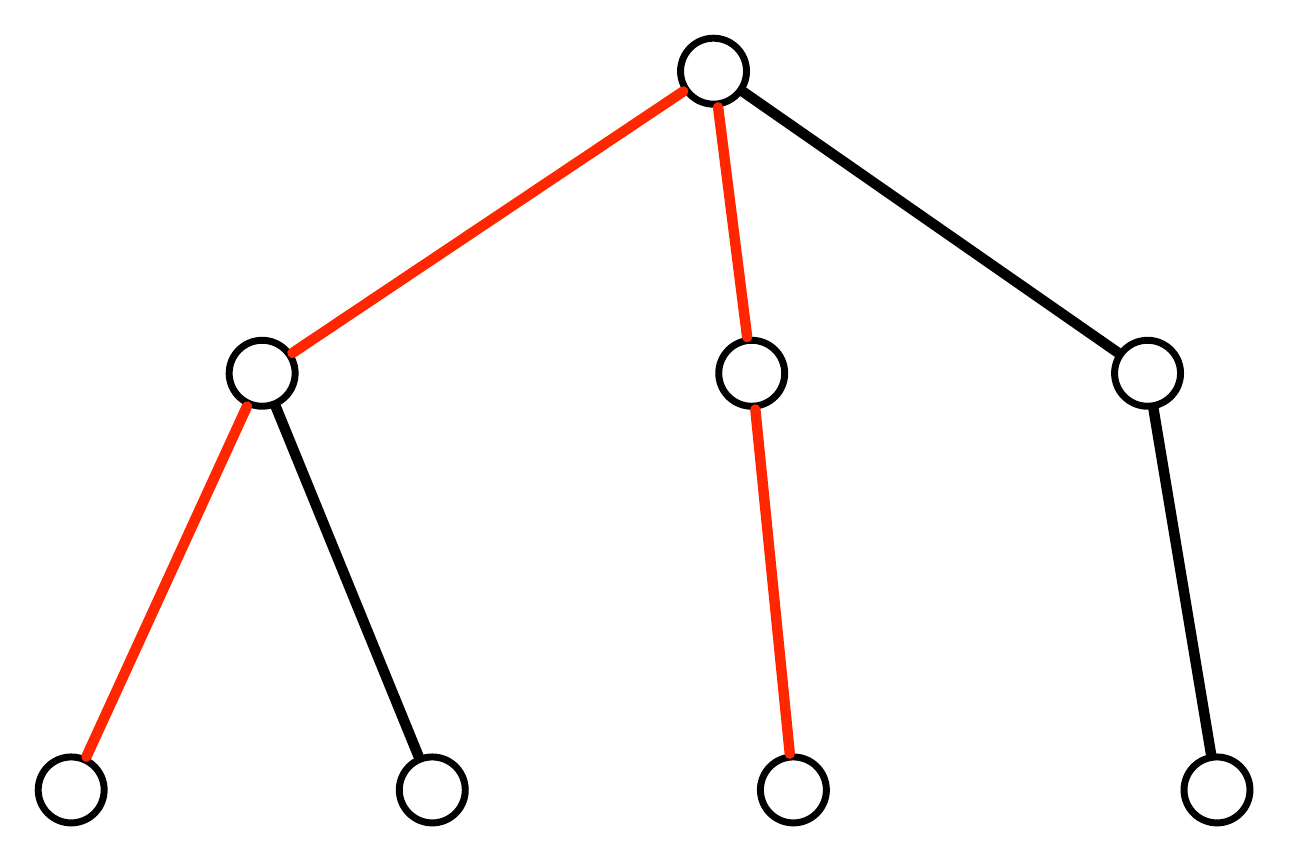}
  \vspace{-0.1in}
  \caption{The red edges indicate a leaf-to-leaf path in the tree. The original edge colors are omitted here for the sake of simplicity.} 
   \label{fig-tree}
  \vspace{-0.15in}
\end{figure}

\noindent\textbf{Case \Rmnum{3}.} The third case is that $G$ only contains the circles having single color. We will show that this case can be handled by a similar way of Case \Rmnum{2}. First, we know that there are the following two different types of vertices in $G$. \textbf{Type $\rmnum{1}$}: the vertex not belonging to any circle; \textbf{Type $\rmnum{2}$}: the vertex belonging to some circle. Due to the construction of $G$ we know that the vertices belonging to a circle actually form a clique, and all of them are type $\rmnum{2}$. So we build a pseudo tree for $G$ recursively as follows. \\

\noindent\textbf{Pseudo-tree(G)}
\begin{enumerate}
\item Initially, pick a vertex $v$ arbitrarily from $G$. 
\item If $v$ is type $\rmnum{1}$, take it as the root. Else, take the whole clique $\mathcal{C}$ containing $v$ as the root.
\item Delete $v$ (if type $\rmnum{1}$) or $\mathcal{C}$ (if type $\rmnum{2}$) and its induced edges. If the remaining of $G$ is not empty, it will become a set of disjoint components $\{G_1, \cdots, G_t\}$.
\item For each component $G_i$, add Pseudo-tree($G_i$) as a child of $v$ or $\mathcal{C}$. 
\end{enumerate}

\begin{figure}[h]
\vspace{-.1cm} 
\centering
 \includegraphics[height=1.3in]{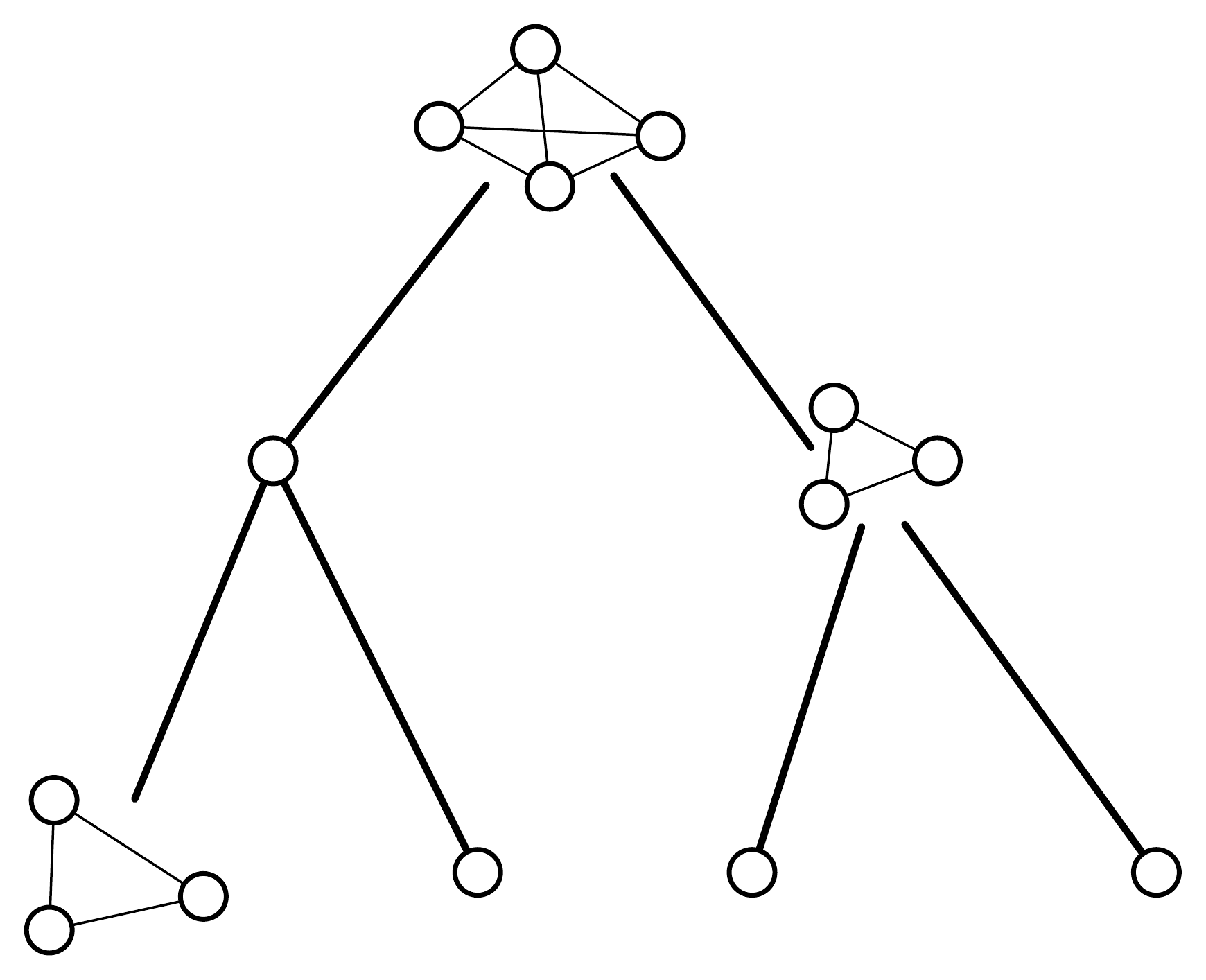}
  \vspace{-0.1in}
  \caption{An example of pseudo tree. The edge colors are omitted here for the sake of simplicity.} 
   \label{fig-ptree}
  \vspace{-0.05in}
\end{figure}

Pseudo-tree(G) returns a pseudo tree where each node is either a type $\rmnum{1}$ vertex or a clique of type $\rmnum{2}$  vertices (see Fig.~\ref{fig-ptree}). Similar to Case \Rmnum{2}, we take an arbitrary leaf-to-leaf path of $G$. If both of the two leaves are type $\rmnum{1}$ vertices, we can adjust the fractional numbers along the path as same as Case \Rmnum{2}, and update $G$ by removing some edges. Otherwise, we focus on the leaf that is a clique $\mathcal{C}$ of type $\rmnum{2}$ vertices. Note that $\mathcal{C}$ contains at least three vertices, and only one of them has an outward edge from the clique (because it is a leaf). Suppose that the two vertices having no outward edge are $v_1$ and $v_2$, and the corresponding two fractional numbers are $x^1_{*_1}$ and $x^2_{*_1}$ respectively. Let $\delta=\min\{x^1_{*_1}-\lfloor x^1_{*_1}\rfloor, \lceil x^2_{*_1}\rceil-x^2_{*1}\}$. Then at least one of 
\begin{eqnarray}
x^1_{*_1}-\delta \hspace{0.2in} and \hspace{0.2in} x^2_{*_1}+\delta
\end{eqnarray}
is an integer, and the other remains non-negative. Similar to Case \Rmnum{2}, we know that all the constraints of the SoL are not violated, and thus an update of $G$ follows. After finite times of such an adjustment, $G$ will become either Case \Rmnum{2} or a graph containing no edge (i.e., an integer solution of the SoL is obtained).

Finally, because the complexity of the initial $G$ is $O(poly(2^k))$, the whole adjustment costs $O(poly(2^k))$ time as well and is independent of $n$.
%
%
%
%
%
%
%
%
%
%
%
%
%
\end{proof}


\subsection{The $4$-Approximation Algorithm}

Combining Section~\ref{sec-candidate} \& \ref{sec-feasible}, we have Algorithm~\ref{alg-1}. Step 1 \& 2 cost $O(knd+nk\log (nk))$ time, and step 3 runs at most $O(k^k\log n)$ rounds where each round costs $O(n+poly(2^k))$ times. Thus, the total running time is $O(n(\log n +d))$ if $k$ is a constant.

\begin{algorithm}[tb]
   \caption{$4$-Approximation Algorithm}
   \label{alg-1}
\begin{algorithmic}
   \STATE {\bfseries Input:} $P=\{p_i, \mid 1\leq i\leq n\}\subset\mathbb{R}^d$, an integer $k\geq 1$, and integer lower and upper bounds $1\leq L\leq U\leq n$.
\begin{enumerate}
\item Run Gonazlez's algorithm and output $k$ points $S=\{s_1, s_2, \cdots, s_k\}$. 
\item Compute the $nk$ distances from $P$ to $S$, and sort them in an increasing order. The set of distances is denoted as $\mathcal{R}$. Initialize the optimal radius $r_{opt}=\max\mathcal{R}$.
\item For each $k$-tuple $(s'_1, \cdots, s'_k)$ from $S^k$, binary search on $\mathcal{R}$. For each step with $r\in \mathcal{R}$, do the following steps.
\begin{enumerate}
\item Draw the $k$ balls with radii $r$ and centered at $(s'_1, \cdots, s'_k)$ separately. 
\item If the SoL is feasible, 
\begin{itemize}
\item update $r_{opt}$ to be $r$  and record the feasible solution if $r<r_{opt}$; 
\item if $r$ is not a leaf, continue the binary search to the left side. Else, stop binary search.
\end{itemize}
\item Else, 
\begin{itemize}
\item if $r$ is not a leaf, continue the binary search to the right side. Else, stop binary search.
\end{itemize}
\end{enumerate}
\item Return the $k$-tuple from $S^k$ with the smallest $r_{opt}$ associating the corresponding feasible solution.
\end{enumerate}
\end{algorithmic}
\end{algorithm}

\begin{theorem}
\label{the-4app}
Algorithm~\ref{alg-1} yields a $4$-approximation of balanced $k$-center clustering, and the running time is $O(n(\log n+d))$ when $k$ is a constant. 
\end{theorem}

\begin{cor}
\label{cor-4app}
Suppose the given instance locates in a metric space, and the time complexity for acquiring the distance between any two points is $O(D)$. Algorithm~\ref{alg-1} yields a $4$-approximation of balanced $k$-center clustering, and the running time is $O(n(\log n+D))$ when $k$ is a constant. 
\end{cor}

\section{Other Issues}
Finally, we address two questions: (1) is the approximation ratio $4$ tight enough, and (2) why should we use $S^k$ rather than $S$ directly?

For the first question, we consider the following example. Let $n=6$ points locate on a line, $k=3$, and $L=U=2$. See Fig.~\ref{fig-4app}. It is easy to know that the optimal solution is $C_1=\{p_1, p_2\}$, $C_2=\{p_3, p_4\}$, and $C_3=\{p_5, p_6\}$ with $r_{opt}=1$. Suppose that the first point selected by Gonazlez's algorithm is $p_2$, then the induced $S=\{p_2, p_5, p_1\}$ which results in a $(4-\delta)$-approximation, no matter which $3$-tuple is chosen from $S^3$. Since $\delta$ can be arbitrarily small, the approximation ratio $4$ is tight.

\begin{figure}[h]
\vspace{-.2cm} 
\centering
 \includegraphics[height=1in]{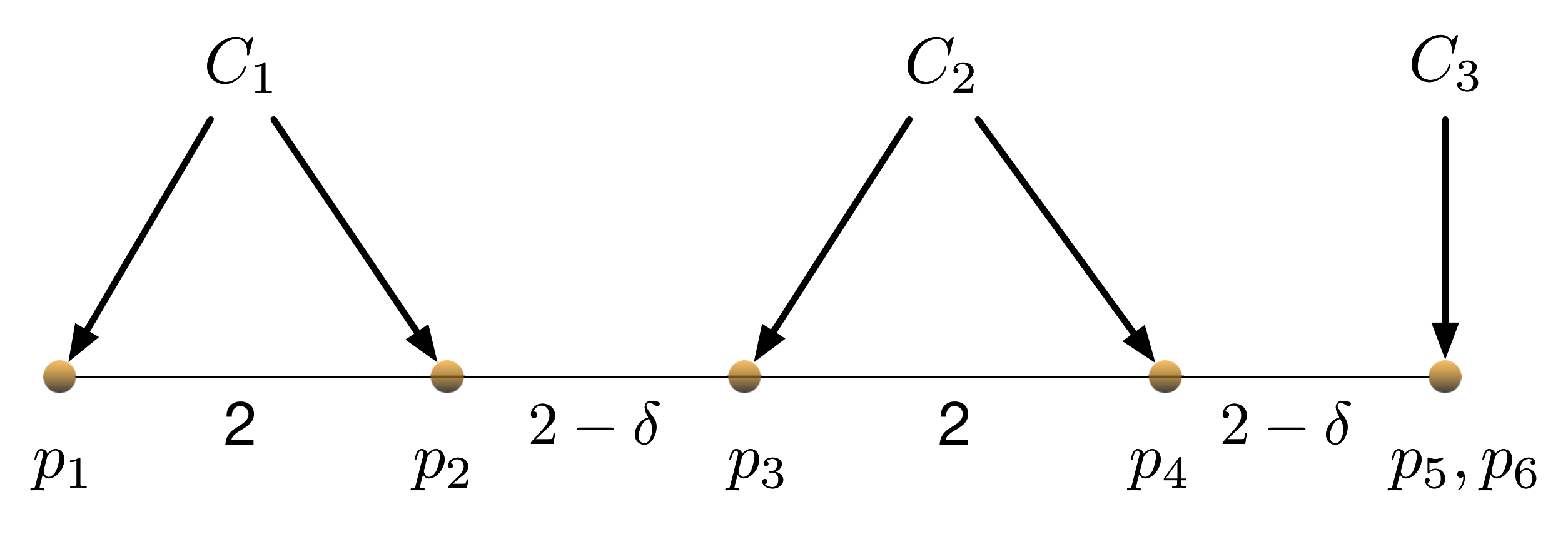}
  \vspace{-0.1in}
  \caption{$||p_1-p_2||=||p_3-p_4||=2$ and $||p_2-p_3||=||p_4-p_5||=2-\delta$ with a small positive $\delta$; $p_5$ and $p_6$ overlap.} 
   \label{fig-4app}
  \vspace{-0.15in}
\end{figure}

We construct another example to answer the second question. See Fig.~\ref{fig-tuple}. It is easy to know $r_{opt}=r$. Suppose that the first point selected by Gonazlez's algorithm is $p_1$, then the induced $S=\{p_1, p_5, p_6\}$. If we take these 3 points as the cluster centers, the obtained radius is at least $h$ (since $p_3$ and $p_4$ have to be assigned to $p_6$). Consequently, the approximation ratio is $h/r$ which can be arbitrarily large. Hence we need to search the $k$-tuple points from $S^k$ rather than $S$. 

\begin{figure}[h]
\vspace{-.2cm} 
\centering
 \includegraphics[height=2in]{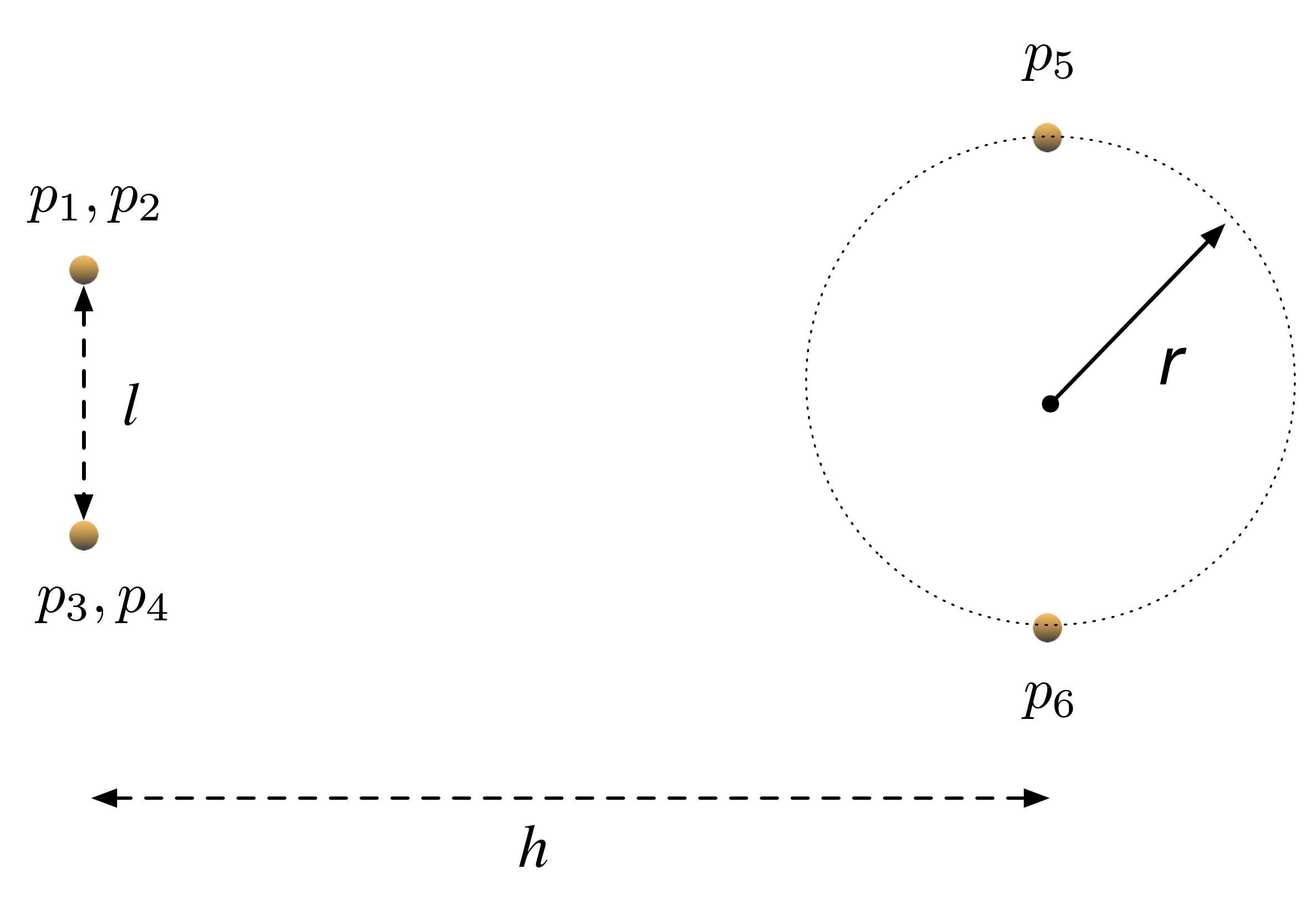}
  \vspace{-0.1in}
  \caption{Let the 6 points locate in a plane, $k=3$, and $L=U=2$. $p_1$ and $p_2$ overlap, $p_3$ and $p_4$ overlap, and these 4 points locate on the same vertical line while $p_5$ and $p_6$ locate on another vertical line; $||p_1-p_3||=l$, $||p_5-p_6||=2r$, and their horizontal distance is $h$; $l<2r\ll h$.} 
   \label{fig-tuple}
  \vspace{-0.15in}
\end{figure}

%
%
%
%
%
%
%
%
%
%
%
%
%
%
%
%
%
%


\small
\bibliographystyle{abbrv}

\begin{thebibliography}{99}

\bibitem{APF}
G. Aggarwal,  R. Panigrahy, T. Feder, D. Thomas, K.  Kenthapadi, S. Khuller,
and A. Zhu, Achieving anonymity via clustering. ACM Trans. Algorithms 6(3),
49:1-49:19 (2010).

\bibitem{AS16}
S. Ahmadian and C. Swamy,  Approximation algorithms for clustering problems with
lower bounds and outliers. arXiv preprint arXiv:1608.01700 (2016).


\bibitem{ABC}
H. C. An, A. Bhaskara, C. Chekuri, S. Gupta, V. Madan, and O. Svensson,  Centrality
of trees for capacitated k-center. Mathematical Programming 154(1-2), 29-53 (2015).



\bibitem{ABM}
K. Aydin, M. Bateni, and V. S. Mirrokni, Distributed Balanced Partitioning via Linear Embedding. WSDM 2016: 387-396.





\bibitem{BKP}
J. Barilan, G. Kortsarz, and D. Peleg,  How to allocate network centers. Journal of
Algorithms 15(3), 385-415 (1993).



\bibitem{BBL}
M. Bateni, A. Bhaskara, S. Lattanzi, and V. S. Mirrokni, Distributed Balanced Clustering via Mapping Coresets. NIPS 2014: 2591-2599.





\bibitem{CHK}
M. Cygan, M. Hajiaghayi, and S. Khuller, Lp rounding for k-centers with non-uniform
hard capacities. 53rd Annual Symposium on
Foundations of Computer Science. pp. 273-282. IEEE Computer Society (2012).


\bibitem{DLP}
T. Dick, M. Li, V. Pillutla, C. White, M.-F. Balcan, and A. J. Smola, Data Driven Resource Allocation for Distributed Learning. CoRR abs/1512.04848 (2015).

\bibitem{DHH}
H. Ding, L. Hu, L. Huang, and J. Li, Capacitated Center Problems with Two-Sided Bounds and Outliers. CoRR abs/1702.07435 (2017).


\bibitem{EHR}
A. Ene, S. Har-Peled, and B. Raichel, Fast clustering with lower bounds, No customer
too far, no shop too small. arXiv preprint arXiv:1304.7318 (2013).







\bibitem{G85} 
T. Gonzalez, Clustering to minimize the maximum intercluster distance. Theoret. Comput.
Sci., 38:293-306, 1985.

\bibitem{HS85}
D. S. Hochbaum and D. B. Shmoys, A best possible heuristic for the k-center problem.
Mathematics of operations research 10(2), 180-184 (1985).


\bibitem{KR90}
L. Kaufman and P. J. Roussew, Finding Groups in Data - An Introduction to Cluster Analysis. A Wiley-Science Publication John Wiley \& Sons, 1990.


\bibitem{KS00}
S. Khuller and Y. J. Sussmann, The capacitated k-center problem. SIAM Journal on
Discrete Mathematics 13(3), 403-418 (2000).

\bibitem{KC14}
T. Kociumaka and M. Cygan,  Constant factor approximation for capacitated k-center
with outliers. arXiv preprint arXiv:1401.2874 (2014).


\bibitem{O13}
J. B. Orlin, Max flows in O(nm) time, or better. STOC 2013: 765-774.


\end{thebibliography}


\end{document}